\newif\ifDRAFT \DRAFTtrue    
\ifDRAFT
\documentclass[12pt,a4paper,draftcls,onecolumn]{IEEEtran}
\normalsize
\else
\documentclass[10pt, conference]{IEEEtran}
\renewcommand{\baselinestretch}{.985}
\fi
\pdfoutput=1

\usepackage{epsfig}
\usepackage{psfrag}
\usepackage{amssymb, amsmath,mathrsfs,amsfonts}
\usepackage{bm,dsfont}
\usepackage{cite}   
\usepackage[Symbol]{upgreek}
\newcommand{\eref}[1]{(\ref{#1})}                  

\newcommand{\beq}{\begin{equation}}
\newcommand{\eeq}{\end{equation}}

\newcommand{\beqa}{\vspace{0mm} \begin{eqnarray}}
\newcommand{\eeqa}{\end{eqnarray} \vspace{0mm}}

\newcommand{\bi}{\begin{itemize}}
\newcommand{\ei}{\end{itemize}}

\newcommand{\ben}{\begin{enumerate}}
\newcommand{\een}{\end{enumerate}}

\newcommand{\lefto}{\mathopen{}\left}

\newcommand{\sizeparentheses}[1]{\lefto( #1 \right)}
\newcommand{\sizecurly}[1]{\lefto\{ #1 \right\}}

\newcommand{\mt}{\mathrm{M_T}} 					
\newcommand{\mr}{\mathrm{M_R}} 						
\newcommand{\minant}{\mathrm{m}} 			
\newcommand{\maxant}{\mathrm{M}} 				
\newcommand{\channelH}{{\matc{H}}} 
\newcommand{\channelHw}{\matc{H}_{\mspace{-2.0mu}w}} 
\newcommand{\channelHwbar}{\overline{\matc{H}}_{\mspace{-2.0mu}w}}

\newcommand{\setS}{\mathcal{S}}                                   

\newcommand{\cn}[2]{\mathcal{CN}\negthinspace\left(#1,#2\right)}             
\newcommand{\mean}[1]{\mathbb{E}\negthinspace\left\{#1\right\}}             

\newcommand{\mc}[1]{\mathcal{#1}}

\newcommand{\mat}[1]{\mathbf{#1}}  
  
\newcommand{\matc}[1]{\bv{\mathcal{#1}}}                             
\newcommand{\bv}[1]{\mbox{\boldmath{$#1$}}}                     
\newcommand{\msf}[1]{\mathsf{#1}}													
\newcommand{\mbb}[1]{\mathbb{#1}}													
\newcommand{\rank}[1]{\mathrm{rank}\negthinspace\left(#1\right)}
\newcommand{\vecop}[1]{\mathrm{vec} \negthinspace\left(#1\right)}             

\newcommand{\expf}[1]{\exp{\negmedspace\sizeparentheses{#1}}}                    
\newcommand{\trace}[1]{\mathrm{Tr}\left(#1\right)}              
\newcommand{\diag}[1]{\mathrm{diag}\negthinspace\left\{#1\right\}}             

\newcommand{\jensen}{\mathrm{J}}                               
\newcommand{\mutinfexp}[2]{\log\det\negthinspace\left(\mat{I}_{#1}+ #2 \right)}           
  
\newcommand{\snr}{\msf{SNR}}                                 
\newcommand{\dotleq}{\mathrel{\dot{\leq}}} 
\newcommand{\dotgeq}{\mathrel{\dot{\geq}}} 

\newcommand{\prob}[1]{\mathbb{P}\negthinspace\left(#1\right)}
\newcommand{\set}[1]{\mathcal{S}_{#1}}                          
\newcommand{\outage}{\mathcal{O}}  
\newcommand{\joutage}{\mathcal{J}}

\newcommand{\codebook}{\mathcal{C}}

\usepackage[stretch=16,shrink=16,step=4]{microtype}

\makeatletter
\def\@IEEEinterspaceratioM{0.265}
\def\@IEEEinterspaceMINratioM{0.1651}
\def\@IEEEinterspaceMAXratioM{0.38}

\def\@IEEEinterspaceratioB{0.31}
\def\@IEEEinterspaceMINratioB{0.19}
\def\@IEEEinterspaceMAXratioB{0.38}
\@IEEEtunefonts
\makeatother

\newtheorem{tc}{Theorem}

\newtheorem{cor}{Corollary}

\providecommand{\norm}[1]{\left\lVert#1\right\rVert}
\providecommand{\abs}[1]{\left\lvert#1\right\rvert}
\begin{document}
\title{Selective-Fading Multiple-Access MIMO Channels: Diversity-Multiplexing Tradeoff and Dominant Outage Event Regions}
%
\author{\authorblockN{Pedro Coronel, Markus G\"artner, and Helmut B\"olcskei\vspace{.5mm}\\}
\authorblockA{Communication Technology Laboratory\\ 
ETH Zurich, 8092 Zurich, Switzerland \\
E-mail: \{pco, gaertner, boelcskei\}@nari.ee.ethz.ch}
\thanks{This work was supported in part by the Swiss National Science Foundation (SNF) under grant No. 200020-109619 and by the STREP project No. IST-026905 MASCOT within the Sixth Framework Programme of the European Commission. Parts of this work were presented at the \textit{IEEE Int. Symp. Inf. Theory (ISIT)}, Toronto, ON, Canada, July 2008.}
}

\IEEEoverridecommandlockouts

\maketitle
\ifDRAFT
\renewcommand{\baselinestretch}{2.0}
\normalsize
\fi

\vspace{-1cm}
\begin{abstract}

We establish the optimal diversity-multiplexing (DM) tradeoff for coherent selective-fading multiple-access MIMO channels and provide corresponding code design criteria. As a byproduct, on the conceptual level, we find an interesting relation between the DM tradeoff framework and the notion of dominant error event regions, first introduced in the AWGN case by Gallager, \textit{IEEE Trans. IT}, 1985. This relation allows us to accurately characterize the error mechanisms in MIMO fading multiple-access channels. In particular, we find that, for a given rate tuple, the maximum achievable diversity order is determined by a single outage event that dominates the total error probability exponentially in SNR. Finally, we examine the distributed space-time code construction 
proposed by Badr and Belfiore, \textit{Int. Zurich Seminar on Commun.}, 2008, using
the code design criteria derived in this paper. \looseness-1
\end{abstract}

\section{Introduction}

The diversity-multiplexing (DM) tradeoff framework introduced by Zheng and Tse allows to efficiently characterize the information-theoretic performance limits of communication over multiple-input multiple-output (MIMO) fading channels both in the point-to-point \cite{ZheTse02} and in the multiple-access (MA) case \cite{TVZ04}. For coherent point-to-point flat-fading channels, DM tradeoff optimal code constructions have been reported in \cite{YaoWor03, GamCaiDam04, BelRekVit05 ,TW05}. The optimal DM tradeoff in point-to-point selective-fading MIMO channels was characterized in \cite{pco07}. In the MA case, the optimal DM tradeoff is known only for flat-fading channels \cite{TVZ04}. A corresponding DM tradeoff optimal code construction was reported in \cite{NamGam07}.

\textit{Contributions.} The aim of this paper is to characterize the DM tradeoff in selective-fading MIMO multiple-access channels (MACs) and to derive corresponding code design criteria. As a byproduct, on the conceptual level, we find an interesting relation between the DM tradeoff framework and the notion of dominant error event regions, first introduced in the case of additive white Gaussian noise (AWGN) MACs by Gallager \cite{Gallager85}. This relation leads to an accurate characterization of the error mechanisms in MIMO fading MACs. Furthermore, we extend the techniques introduced in \cite{pco07} for computing the DM tradeoff in point-to-point selective-fading channels to the MA case. Finally, we examine the distributed space-time code construction proposed in \cite{BadBel08} using the code design criteria derived in this paper.

\textit{Notation.} $\mt$ and $\mr$ denote, respectively, the number of transmit antennas for each user and the number of receive antennas. The set of all users is $\mc{U}=\{1,2,\ldots, U\}$, $\setS$ is a subset of $\mc{U}$ with $\bar{\setS}$ and $|\set{}|$ denoting its complement in $\mc{U}$ and its cardinality, respectively. The superscripts ${}^T$ and ${}^H$ stand for transposition and conjugate transposition, respectively. $\mat{A}\otimes \mat{B}$ and $\mat{A}\odot \mat{B}$ denote, respectively, the Kronecker and Hadamard products of the matrices $\mat{A}$ and $\mat{B}$. If $\mat{A}$ has the columns $\mat{a}_k$ ($k \negmedspace = \negmedspace 1, 2, \ldots, m$), $\vecop{\mat{A}}=[\mat{a}_1^T \:  \mat{a}_2^T \: \cdots \: \mat{a}_m^T]^T$. $\norm{\mat{a}}$ and $\norm{\mat{A}}_\mathrm{F}$ denote, respectively, the Euclidean norm of the vector $\mat{a}$ and the Frobenius norm of the matrix $\mat{A}$. For index sets $\mc{I}_1 \subseteq \sizecurly{1, 2, \ldots, n}$ and $\mc{I}_2 \subseteq \sizecurly{1, 2, \ldots, m}$, $\mat{A}(\mc{I}_{1},\mc{I}_{2})$ stands for the (sub)matrix consisting of the rows of $\mat{A}$ indexed by $\mc{I}_{1}$ and the columns of $\mat{A}$ indexed by $\mc{I}_{2}$. The eigenvalues of the $n\times n$ Hermitian matrix $\mat{A}$, sorted in ascending order, are denoted by $\lambda_k(\mat{A})$, $k=1,2,\ldots, n$. The Kronecker delta function is defined as $\delta_{n,m}=1$ for $n=m$ and zero otherwise. If  $X$ and $Y$ are random variables (RVs), $X\sim Y$ denotes equivalence in distribution and $\mathbb{E}_X$ is the expectation operator with respect to (w.r.t.) the RV $X$. The random vector $\mat{x} \sim \cn{\mat{0}}{\mat{C}}$ is zero-mean jointly proper Gaussian (JPG) with $\mean{\mat{x}\mat{x}^H}=\mat{C}$. $f(x)$ and $g(x)$ are said to be exponentially equal, denoted by $f(x)\doteq g(x)$, if $\lim_{x\rightarrow \infty} \frac{\log f(x)}{\log x} = \lim_{x\rightarrow \infty} \frac{\log g(x)}{\log x}$. Exponential inequality, indicated by $\:\dotgeq$ and $\dotleq$, is defined analogously. 

\section{Channel and signal model\label{Sec.Model}}

We consider a selective-fading MAC where $U$ users, with $\mt$ transmit antennas each, communicate with a single receiver with $\mr$ antennas. The corresponding input-output relation is given by \looseness-1
\beqa
\mat{y}_n = \sqrt{\frac{\snr}{\mt}} \sum_{u=1}^U \mat{H}_{u,n} \:\mat{x}_{u, n} +\mat{z}_n,\: n=0, 1, \ldots, N-1,\label{Eq.SigModel}
\eeqa
where the index $n$ corresponds to a time, frequency, or time-frequency slot and SNR denotes the per-user signal-to-noise ratio at each receive antenna. The vectors $\mat{y}_n $, $\mat{x}_{u,n}$, and  $\mat{z}_n$ denote, respectively, the $\mr \times 1$ receive signal vector, the $\mt \times 1$ transmit signal vector corresponding to the $u$th user, and the $\mr \times 1$ zero-mean JPG noise vector satisfying $\mean{\mat{z}_n\mat{z}_{n' }^H}= \delta_{n,n'}\: \mat{I}_{\mr}$, all for the $n$th slot. We assume that the receiver has perfect knowledge of all channels and the transmitters do not have channel state information (CSI) but know the channel law.

We restrict our analysis to spatially uncorrelated Rayleigh fading channels so that, for a given $n$, $\mat{H}_{u,n}$ has i.i.d. $\cn{0}{1}$ entries. The channels corresponding to different users are assumed to be statistically independent. We do, however, allow for correlation across $n$ for a given $u$, and assume, for simplicity, that all scalar subchannels have the same correlation function so that $\mean{\mat{H}_{u,n}(i,j)\: (\mat{H}_{u',n'}(i,j))^*}=\mat{R}_\mbb{H}(n,n') \delta_{u,u'}$, for  $i=1, 2, \ldots, \mr$ and $ j=1, 2, \ldots, \mt$. The covariance matrix $\mat{R}_\mbb{H}$ is obtained from the channel's time-frequency correlation function \cite{Bello63}. In the sequel, we let $\rho\triangleq\rank{\mat{R}_\mbb{H}}$. For any set $\setS=\{u_1, \ldots, u_{|\setS|}\}$, we stack the corresponding users' channel matrices for a given slot index $n$ according to \looseness-1
\begin{equation}
\mat{H}_{\setS,n} = [\mat{H}_{u_1,n} \: \cdots \: \mat{H}_{u_{|\setS|},n}].\label{Eq.Hsn}
\end{equation}
With this notation, it follows that
\beqa
\mean{\vecop{\mat{H}_{\setS,n}}(\vecop{\mat{H}_{\setS,n'}})^H} = \mat{R}_\mbb{H}(n,n') \;\mat{I}_{|\setS|\mt\mr}\label{Eq.CovModel}.
\eeqa

\section{Preliminaries}

Assuming that all users employ i.i.d. Gaussian codebooks\footnote{A standard argument along the lines of that used to obtain \cite[Eq. 9]{ZheTse02} shows that this assumption does not entail a loss of optimality in the high SNR regime, relevant to the DM tradeoff.}, the set of achievable rate tuples $(R_1,R_2, \ldots, R_U)$ for a given channel realization $\{\mat{H}_{\mc{U},n}\}_{n=0}^{N-1}$ is given by
\begin{equation}
\begin{split}
 \mc{R}&= \Bigg\{(R_1, R_2, \ldots, R_U): \forall \setS \subseteq\mathcal{U}, \\&R(\setS)\leq \frac{1}{N} \sum_{n=0}^{N-1}\mutinfexp{}{\frac{\snr}{\mt}\:\mat{H}_{\setS,n}\mat{H}_{\setS,n}^H}  \Bigg\}\label{Eq.ART}
\end{split}
\end{equation}
where $R(\setS) = \sum_{u\in \setS} R_u$. If a given rate tuple $(R_1, R_2, \ldots, R_U)\notin \mc{R}$, we say that the channel is in outage w.r.t. this rate tuple. Denoting the corresponding outage event as $\outage$, we have
\begin{equation}
\prob{\outage} = \prob{\bigcup_{\setS\:\subseteq\:\mathcal{U}} \outage_\setS} \label{Eq.P1} 
\end{equation}
where the $\setS$-outage event $\outage_\setS$ is defined as
\begin{equation}\label{Eq.Os}
\begin{split}
\outage_\setS &\triangleq \Bigg\{\{\mat{H}_{\setS,n}\}_{n=0}^{N-1}: \\ &\hspace{3mm}\frac{1}{N} \sum_{n=0}^{N-1}\mutinfexp{}{\frac{\snr}{\mt}\: \mat{H}_{\setS,n}\mat{H}_{\setS,n}^H} < R(\setS)\Bigg\}.
\end{split}
\end{equation}

Our goal is to characterize \eref{Eq.P1} as a function of the rate tuple $(R_1, R_2, \ldots, R_U)$ in the high-SNR regime and to establish sufficient conditions on the users' codebooks to guarantee that the corresponding error probability is exponentially (in SNR) equal to $\prob{\outage}$. To this end, we employ the DM tradeoff framework \cite{ZheTse02}, which, in its MA version \cite{TVZ04}, will be briefly summarized next.\looseness-1

In the DM tradeoff framework, the data rate of user $u$ scales with SNR as $R_u(\snr) = r_u \log\snr$, where $r_u$ denotes the multiplexing rate. Consequently, a sequence of codebooks $\codebook_{r_u}(\snr)$, one for each SNR, is required. We say that this sequence of codebooks constitutes a family of codes $\codebook_{r_u}$ operating at multiplexing rate $r_u$. The family $\codebook_{r_u}$ is assumed to have block length $N$. At any given SNR, $\codebook_{r_u}(\snr)$ contains codewords $\mat{X}_u= [\mat{x}_{u,0}\: \mat{x}_{u,1}\: \cdots \:\mat{x}_{u,N-1}]$ satisfying the per-user power constraint \looseness-1
\begin{equation}
\norm{\mat{X}_u}_\mathrm{F}^2 \leq \mt N, \; \forall\:\mat{X}_u \in \codebook_{r_u}.\label{Eq.PC}
\end{equation}
In the remainder of the paper, we will say ``the power constraint \eref{Eq.PC}'' to mean that \eref{Eq.PC} has to be satisfied for $u=1,2,\ldots, U$. The overall family of codes is given by $\codebook_{\bm{r}}=\codebook_{r_1}\times \codebook_{r_2}\times \cdots \times \codebook_{r_U}$, where $\bm{r}=(r_1, r_2, \ldots, r_U)$ denotes the multiplexing rate tuple\footnote{Throughout the paper, we consider multiplexing rate tuples lying within the boundaries determined by the ergodic capacity region.}. At a given SNR, the corresponding codebook $\codebook_{\bm{r}}(\snr)$ contains $\snr^{Nr(\mc{U})}$ codewords with $r(\mc{U})=\sum_{u=1}^U r_u$. \looseness-1

The DM tradeoff realized by $\codebook_{\bm{r}}$ is characterized by the function 
\begin{equation*}
d(\codebook_{\bm{r}})=-\lim_{\snr\rightarrow\infty}\frac{\log P_e(\codebook_{\bm{r}})}{\log\snr}
\end{equation*} 
where $P_e(\codebook_{\bm{r}})$ is the \textit{total} error probability obtained through maximum-likelihood (ML) decoding, that is, the probability for the receiver to decode at least one user in error. The optimal DM tradeoff curve $d^\star\mspace{-2.0mu}(\bm{r}) = \sup_{\codebook_{\bm{r}}} d(\codebook_{\bm{r}})$, where the supremum is taken over all possible families of codes satisfying the power constraint \eref{Eq.PC}, quantifies the maximum achievable diversity order as a function of the multiplexing rate tuple $\bm{r}$. Since the outage probability $\prob{\outage}$ is a lower bound (exponentially in SNR) on the error probability of any coding scheme \cite[Lemma 7]{TVZ04}, we have\looseness-1
\begin{equation}
d^\star\mspace{-2.0mu}(\bm{r}) \leq - \lim_{\snr\rightarrow\infty}\frac{\log\prob{\outage}}{\log\snr}\label{Eq.OutLB}
\end{equation}
where the outage event $\outage$, defined in \eref{Eq.P1} and \eref{Eq.Os}, is w.r.t. the rates $R_u(\snr)=r_u \log\snr$, $\forall u$. As an extension of the corresponding result for the flat-fading case \cite{TVZ04}, we shall show that \eref{Eq.OutLB} holds with equality also for selective-fading MACs. However, just like in the case of point-to-point channels, a direct characterization of the right-hand side (RHS) of \eref{Eq.OutLB} for the selective-fading case seems analytically intractable since one has to deal with the sum of correlated (recall that the $\mat{H}_{u,n}$ are correlated across $n$) terms in \eref{Eq.Os}. In the next section, we show how the technique introduced in \cite{pco07} for characterizing the DM tradeoff of point-to-point selective-fading MIMO channels can be extended to the MA case.\looseness-1

\section{Computing the optimal DM tradeoff curve}

\subsection{Lower bound on $\prob{\outage_\setS}$}

First, we derive a lower bound on the individual terms $\prob{\outage_\setS}$. We start by noting that for any set $\setS\subseteq \mc{U}$, Jensen's inequality provides the following upper bound:
\begin{equation}
\frac{1}{N} \sum_{n=0}^{N-1} \mutinfexp{}{\frac{\snr}{\mt}\mat{H}_{\setS,n}\mat{H}_{\setS, n}^H}
\leq \log\det\sizeparentheses{\mat{I}+\frac{\snr}{\mt N}\channelH_\setS \channelH_\setS^H} \triangleq \jensen_\setS(\snr) \label{Eq.Jensen}
\end{equation}
where the ``Jensen channel" \cite{pco07} is defined as 
\begin{equation}
\channelH_\setS = \begin{cases} [\mat{H}_{\setS ,0} \; \mat{H}_{\setS ,1}\; \cdots \; \mat{H}_{\setS ,N-1}], & \text{if $\mr\leq|\setS|\mt$,}\\ 
[\mat{H}_{\setS ,0}^H \; \mat{H}_{\setS ,1}^H\; \cdots \; \mat{H}_{\setS ,N-1}^H], & \text{if $\mr > |\setS|\mt$.} \end{cases}
\end{equation} 
Consequently, $\channelH_\setS$ has dimension $\minant(\setS) \times N \maxant(\setS)$, where
\begin{align}
\minant(\setS)&\triangleq\min(|\setS|\mt,\mr)\\
\maxant(\setS)&\triangleq\max(|\setS|\mt,\mr).
\end{align} 
In the following, we say that the event $\joutage_\setS$ occurs if the Jensen channel $\channelH_\setS$ is in outage w.r.t. the rate $r(\setS)\log\snr$, where $r(\setS) = \sum_{u\in \setS} r_u$, i.e., $\joutage_\setS\triangleq\sizecurly{\jensen_\setS(\snr)<r(\setS)\log\snr}$. From \eref{Eq.Jensen}, we can conclude that, obviously, $\prob{\joutage_\setS}\leq\prob{\outage_\setS}$. \looseness-1

We next characterize the Jensen outage probability analytically. Recalling \eref{Eq.CovModel}, we start by expressing $\channelH_\setS$ as $\channelH_\setS = \channelHw (\mat{R}^{T/2} \otimes \mat{I}_{\maxant(\setS)})$, where $\mat{R}=\mat{R}_\mbb{H}$, if $\mr \leq |\setS|\mt$, and $\mat{R}=\mat{R}_\mbb{H}^T$, if $\mr > |\setS|\mt$, and $\channelHw$ is an i.i.d. $\cn{0}{1}$ matrix with the same dimensions as $\channelH_\setS$ given by
\begin{equation}
\channelHw = \begin{cases} \; [\mat{H}_{w ,0} \; \mat{H}_{w ,1}\; \cdots \; \mat{H}_{w ,N-1}], & \text{if $\mr\leq|\setS|\mt$,} \\ 
[\mat{H}_{w ,0}^H \; \mat{H}_{w ,1}^H\; \cdots \; \mat{H}_{w ,N-1}^H], & \text{if $\mr > |\setS|\mt$.}\end{cases} \label{Eq.Hwn}
\end{equation}
Here, $\mat{H}_{w ,n}$ denotes i.i.d. $\cn{0}{1}$ matrices of dimension $\mr\times |\setS|\mt$. Using $\channelHw \mat{U} \sim \channelHw$, for any unitary $\mat{U}$, and $\lambda_n(\mat{R}_\mbb{H})=\lambda_n(\mat{R}_\mbb{H}^T)$ for all $n$, we get 
\begin{equation}
\channelH_\setS\channelH_\setS^H \sim \channelHw (\mat{\Lambda}\otimes \mat{I}_{\maxant(\setS)}) \channelHw ^H
\end{equation} 
where $\mat{\Lambda}=\diag{\lambda_1(\mat{R}_\mbb{H}), \lambda_2(\mat{R}_\mbb{H}), \ldots, \lambda_\rho(\mat{R}_\mbb{H}), 0,\ldots, 0}$. Setting $\channelHwbar = \channelHw([1:\minant(\setS)], [1:\rho\maxant(\setS)])$, it was shown in \cite{pco07} that $\prob{\joutage_\setS}$ is nothing but the outage probability of an effective MIMO channel with $\rho\maxant(\setS)$ transmit and $\minant(\setS)$ receive antennas and satisfies\looseness-1
\begin{align}
\prob{\joutage_\setS} &\doteq \prob{\log\det\sizeparentheses{\mat{I}_{}+ \snr \:\channelHwbar\channelHwbar^H}\negmedspace< r(\setS) \log\snr}\notag\\ &\doteq \snr^{-d_\setS(r(\setS))}\label{Eq.Outprob2}
\end{align}
where we infer from the results in \cite{ZheTse02} that $d_\setS(r)$ is the piecewise linear function connecting the points $(r,d_\setS(r))$ for $r=0, 1, \ldots, \minant(\setS)$, with
\begin{equation}
d_\setS(r) = (\minant(\setS)-r)(\rho\maxant(\setS)-r).\label{Eq.JensenCurve}
\end{equation}
Since, as already noted, $\prob{\outage_\setS}\geq\prob{\joutage_\setS}$, it follows from \eref{Eq.Outprob2} that
\begin{equation}
\prob{\outage_\setS} \dotgeq \snr^{-d_\setS(r(\setS))} \label{Eq.LBOs}
\end{equation}
which establishes the desired lower bound.


\subsection{Error event analysis}

Following \cite{Gallager85,TVZ04}, we decompose the total error probability into $2^U-1$ disjoint error events according to
\begin{equation}
P_e(\codebook_{\bm{r}}) = \sum_{\setS\:\subseteq\;\mc{U}} \prob{\mc{E}_\setS}\label{Eq.PeDecomposition}
\end{equation}
where the $\setS$-error event $\mc{E}_\setS$ corresponds to \textit{all} the users in $\setS$ being decoded in error and the remaining users being decoded correctly. More precisely, we have
\begin{equation}\label{Eq.SError}
\mc{E}_\setS \triangleq \sizecurly{(\hat{\mat{X}}_u\neq \mat{X}_u, \forall u \in \setS)\; \land\; (\hat{\mat{X}}_u= \mat{X}_u, \forall u \in \bar{\setS})}
\end{equation}
where ${\mat{X}}_u$ and $\hat{\mat{X}}_u$ are, respectively, the transmitted and ML-decoded codewords corresponding to user $u$. We note that, in contrast to the outage events $\outage_\setS$ defined in \eref{Eq.Os}, the error events $\mc{E}_\setS$ are disjoint. The following result establishes the DM tradeoff optimal code design criterion for a specific error event $\mc{E}_\setS$.\looseness-1

\begin{tc}\label{Th.CDC}
For every $u\in \setS$, let $\codebook_{r_u}$ have block length $N\geq\rho|\setS|\mt$. Let the nonzero\footnote{Here, we actually mean the eigenvalues that are not identically equal to zero for all SNR values. This fine point will be made clear in the proof.} eigenvalues of $\mat{R}_\mbb{H}^T \odot (\sum_{u\in\setS} \mat{E}_u^H\mat{E}_u)$, where $\mat{E}_u=\mat{X}_u-\mat{X}_u'$ and $\mat{X}_u$, $\mat{X}_u' \in \codebook_{r_u}(\snr)$, be given---in ascending order---at every SNR level by $\uplambda_n(\snr)$, $n=1, 2, \ldots, \rho|\setS|\mt$. Furthermore, set
\begin{equation}
\Lambda_{\minant(\setS)}^{\rho|\setS|\mt}(\snr) \triangleq \min_{\begin{subarray}{c} \{\mat{E}_u=\mat{X}_u-\mat{X}_u'\}_{u\in \setS}\\\mat{X}_u,\mat{X}_u' \:\in\: \codebook_{r_u}(\snr)\end{subarray}} \quad \prod_{k=1}^{\minant(\setS)} \uplambda_k(\snr). \label{Eq.DefLambda}
\end{equation}
If there exists an $\epsilon>0$ independent of $\snr$ and $r$ such that
\begin{equation}
\Lambda_{\minant(\setS)}^{\rho|\setS|\mt}(\snr)\dotgeq \snr^{-(r(\setS)-\epsilon)},\label{Eq.ThCDC}
\end{equation} 
then, under ML decoding, $\prob{\mc{E}_\setS} \dotleq \snr^{-d_\setS(r(\setS))}$.
\end{tc}
\begin{proof}
We start by deriving an upper bound on the average (w.r.t. the random channel) pairwise error probability (PEP) of an $\setS$-error event. From \eref{Eq.SError}, we note that $\mat{E}_{u}=[\mat{e}_{u,0} \: \mat{e}_{u,1} \: \cdots \: \mat{e}_{u,N-1}]$, with $\mat{e}_{u,n}=\mat{x}_{u,n}-\mat{x}_{u,n}'$, is nonzero for $u\in \setS$ but $\mat{E}_u=\mat{0}$ for any $u\in \bar{\setS}$. Assuming, without loss of generality, that $\setS=\{1, \ldots, |\setS|\}$, the probability of the ML decoder mistakenly deciding in favor of the codeword $\mat{X}'$ when $\mat{X}$ was actually transmitted can be upper-bounded in terms of $\mat{X}-\mat{X}'= [\mat{E}_{1}\:\cdots\:\mat{E}_{|\setS|}\: \mat{0}\:\cdots\:\mat{0}]$ as\looseness-1
\begin{equation}\begin{split}
&\prob{\mat{X}\rightarrow\mat{X}'}\\
&\negmedspace\leq \mathbb{E}_{\{\mat{H}_{\setS,n}\}_{n=0}^{N-1}}\sizecurly{\negthinspace\expf{\negthinspace-\frac{\snr}{4\mt}\sum_{n=0}^{N-1}\trace{\mat{H}_{\setS,n}\mat{e}_n\mat{e}_n^H\mat{H}_{\setS,n}^H}\negmedspace}\negmedspace}\label{p1}
\end{split}\end{equation}
where 
\begin{equation*}
\trace{\mat{H}_{\setS,n}\mat{e}_n\mat{e}_n^H\mat{H}_{\setS,n}^H}=\norm{\sum_{u\in\setS}\mat{H}_{u,n} \mat{e}_{u,n}}^2
\end{equation*} 
with $\mat{H}_{\setS,n}$ defined in \eref{Eq.Hsn} and $\mat{e}_n=[\mat{e}_{u_1,n}^T\:\cdots\:\mat{e}_{u_{|\setS|},n}^T]^T$. Setting $\mat{H}_{\setS}=[\mat{H}_{\setS,0}\;\mat{H}_{\setS,1}\;\cdots\;\mat{H}_{\setS,N-1}]$, we get from \eref{p1}\looseness-1
\begin{align}
&\prob{\mat{X}\rightarrow\mat{X}'} \notag \\ &\leq \mathbb{E}_{\mat{H}_\setS}\sizecurly{\expf{-\frac{\snr}{4\mt}\trace{\mat{H}_{\setS}\:\diag{\mat{e}_n\mat{e}_n^H}_{n=0}^{N-1}\mat{H}_{\setS}^H}}}\notag\\
&= \mathbb{E}_{\mat{H}_w}\sizecurly{\expf{-\frac{\snr}{4\mt}\trace{\mat{H}_{w}\mat{\Upsilon}\mat{\Upsilon}^H\mat{H}_{w}^H}}}\label{p2}
\end{align}
where we have used $\mat{H}_\setS=\mat{H}_w (\mat{R}_\mbb{H}^{T/2} \otimes \mat{I}_{|\setS|\mt})$ with $\mat{H}_w$ an $\mr \times N |\setS|\mt$ matrix with i.i.d. $\cn{0}{1}$ entries and 
\begin{equation}
\mat{\Upsilon}=(\mat{R}_{\mbb{H}}^{T/2} \otimes \mat{I}_{|\setS|\mt}) \:\diag{\mat{e}_n}_{n=0}^{N-1}.\label{Eq.WhiteMat}
\end{equation}
We note that $\mat{\Upsilon}^H\mat{\Upsilon}=\mat{R}_\mbb{H}^T \odot (\sum_{u\in\setS}\mat{E}_u^H\mat{E}_u)$, where we have $\rank{\sum_{u\in\setS}\mat{E}_u^H\mat{E}_u}\leq |\setS|\mt$ because $\mat{E}_u$ has dimension $\mt\times N$ and $N\geq |\setS|\mt$ by assumption. Recalling that $\rank{\mat{R}_\mbb{H}}=\rho$ and using the property $\rank{\mat{A}\odot\mat{B}}\leq\rank{\mat{A}}\rank{\mat{B}}$, it follows that $\rank{\mat{\Upsilon}^H\mat{\Upsilon}}\leq\rho|\setS|\mt$, which is to say that $\mat{\Upsilon}^H\mat{\Upsilon}$ has at most $\rho|\setS|\mt$ eigenvalues that are not identically equal to zero for all SNRs. We stress, however, that these eigenvalues may decay to zero as a function of SNR. Next, using the fact that for any matrix $\mat{A}$ the nonzero eigenvalues of $\mat{A}\mat{A}^H$ equal the nonzero eigenvalues of $\mat{A}^H\mat{A}$, the assumption (made in the statement of the theorem) that $\mat{R}_\mbb{H}^T \odot (\sum_{u\in\setS}\mat{E}_u^H\mat{E}_u)$ has $\rho|\setS|\mt$ eigenvalues that are not identically equal to zero for all SNRs implies that so does $\mat{\Upsilon \Upsilon}^H$. The remainder of the proof proceeds along the lines of the proof of \cite[Th. 1]{pcoj07}. In particular, we split and subsequently bound the $\setS$-error probability as
\begin{align}
\prob{\mc{E}_\setS} &= \prob{\mc{E}_\setS, \joutage_\setS} + \prob{\mc{E}_\setS, \bar{\joutage}_\setS}\notag\\
&= \prob{\joutage_\setS} \underbrace{\prob{\mc{E}_\setS| \joutage_\setS}}_{\leq 1} \notag\\
&{\hspace{7mm}}+ \underbrace{\prob{\bar{\joutage}_\setS}}_{\leq 1} \prob{\mc{E}_\setS| \bar{\joutage}_\setS}\notag\\
&\leq  \prob{\joutage_\setS} +\prob{\mc{E}_\setS | \bar{\joutage}_\setS}.\label{Eq.UpBoundErrorProb}
\end{align}
As detailed in the proof for the point-to-point case given in \cite{pcoj07}, the code design criterion \eref{Eq.ThCDC} yields the following upper bound on the second term in \eref{Eq.UpBoundErrorProb}:
\begin{align}\label{Eq.UnionBound}
\prob{\mc{E}_\setS| \bar{\joutage}_\setS} &\dotleq \snr^{Nr(\setS)} \expf{- \frac{\snr^{\epsilon/{\minant(\setS)}}}{4\mt}}.
\end{align}

In contrast to the Jensen outage probability which satisfies $\prob{\joutage_\setS} \doteq \snr^{-d_\setS(r(\setS))}$, the RHS of \eref{Eq.UnionBound} decays exponentially in SNR. Hence, upon inserting \eref{Eq.UnionBound} into \eref{Eq.UpBoundErrorProb}, we get $\prob{\mc{E}_\setS}\dotleq \prob{\joutage_\setS}$, and can therefore conclude that $\prob{\mc{E}_\setS} \dotleq \snr^{-d_\setS(r(\setS))}$.
\end{proof}

In summary, for every $\mc{E}_\setS$, \eref{Eq.ThCDC} constitutes a sufficient condition on $\{\codebook_{r_u}: u\in \setS\}$ for $\prob{\mc{E}_\setS}$ to be exponentially upper-bounded by $\prob{\joutage_\setS}$. This condition is nothing but the DM tradeoff optimal code design criterion for a point-to-point channel with $|\setS|\mt$ transmit antennas and $\mr$ receive antennas presented in \cite{pcoj07}. In order to satisfy this condition, the users' codebooks have to be designed jointly. We stress, however, that this does not require cooperation among users at the time of communication. We are now ready to establish the optimal DM tradeoff for the selective-fading MAC and provide corresponding design criteria on the overall family of codes $\codebook_{\bm{r}}$.\looseness-1

\subsection{Optimal code design}

We start by noting that \eref{Eq.P1} implies $\prob{\outage} \geq \prob{\outage_\setS}$ for any $\setS\subseteq \mc{U}$, which combined with \eref{Eq.LBOs} gives rise to $2^U-1$ lower bounds on $\prob{\outage}$. For a given multiplexing rate tuple $\bm{r}$, the tightest lower bound (exponentially in SNR) corresponds to the set $\setS$ that yields the smallest SNR exponent $d_\setS(r(\setS))$. More precisely, the tightest lower bound is \looseness-1
\begin{equation}
\prob{\outage} \dotgeq \snr^{-d_{\setS^\star}(r(\setS^\star))} \label{Eq.LB}
\end{equation}
with the dominant outage event given by $\outage_{\setS^\star}$, where
\begin{equation}
\setS^\star \triangleq \underset{\setS\:\subseteq\:\mc{U}}{\arg\min} \; d_\setS(r(\setS))\label{Eq.TEE}
\end{equation}
is the dominant outage set. Next, we show that, for any multiplexing rate tuple, the total error probability $P_e(\codebook_{\bm{r}})$ can be made exponentially equal to the RHS of \eref{Eq.LB} by appropriate design of the users' codebooks. As a direct consequence thereof, using \eref{Eq.OutLB}, \eref{Eq.LB}, and $P_e(\codebook_{\bm{r}}) \dotgeq \prob{\outage}$ \cite[Lemma 7]{TVZ04}, we then obtain that $d_{\setS^\star}(r(\setS^\star))$ constitutes the optimal DM tradeoff of the selective-fading MIMO MAC. Before presenting this result, let us define the function $r_{\scriptscriptstyle\setS}(d)$ as the inverse of $d_\setS(r)$, i.e., $d=d_\setS\big(r_{\scriptscriptstyle\setS}(d)\big)$ and $r=r_\setS\big(d_{\scriptscriptstyle\setS}(r)\big)$. We note that $r_{\scriptscriptstyle\setS}(d)$ is a decreasing function of $d$ and $d_\setS(r)$ is a decreasing function of $r$.

\begin{tc}\label{Th.Proc}
The optimal DM tradeoff of the selective-fading MIMO MAC in \eref{Eq.SigModel} is given by $d^\star(\bm{r})=d_{\setS^\star}(r(\setS^\star))$, that is
\begin{equation}\label{Eq.ThOptDMT}
d^\star(\bm{r})=(\minant(\setS^\star)-r(\setS^\star))(\rho\maxant(\setS^\star)-r(\setS^\star)).
\end{equation}
Moreover, if the overall family of codes $\codebook_{\bm{r}}$ satisfies \eref{Eq.ThCDC} for the dominant outage set $\setS^\star$ and, for every $\setS\neq\setS^\star$, there exists $\epsilon>0$ such that
\begin{equation}\label{Eq.NewCDC}
\Lambda_{\minant(\setS)}^{\rho|\setS|\mt}(\snr)\dotgeq \snr^{-(\gamma_\setS-\epsilon)}
\end{equation}
where
\begin{equation}
0\leq \gamma_\setS \leq r_\setS(d_{\setS^\star}(r(\setS^\star)))\label{Eq.gamma}
\end{equation}
then 
\begin{equation}\label{Eq.ThOptCode}
d(\codebook_{\bm{r}}) =d^\star(\bm{r}).
\end{equation}
\end{tc}

\begin{proof}
Using \eref{Eq.PeDecomposition}, we write
\begin{align}\label{tmp1}
P_e(\codebook_{\bm{r}}) = \prob{\mc{E}_{\setS^\star}} + \sum_{\setS\neq\setS^\star} \prob{\mc{E}_{\setS}}.
\end{align}
We bound the terms in the sum on the RHS of \eref{tmp1} separately. By assumption, $\codebook_{\bm{r}}$ satisfies \eref{Eq.ThCDC} for $\setS^\star$ and, hence, it follows from Theorem \ref{Th.CDC} and \eref{Eq.Outprob2} that 
\begin{equation}\label{tmp2}
\prob{\mc{E}_{\setS^\star}} \dotleq \snr^{-d_{\setS^\star}(r(\setS^\star))} \doteq \prob{\joutage_{\setS^\star}}.
\end{equation}
Next, we consider the terms $\prob{\mc{E}_\setS}$ for $\setS\neq \setS^\star$ and use \eref{Eq.UpBoundErrorProb} to write
\begin{align}
\prob{\mc{E}_\setS}&\leq \prob{\joutage_\setS} + \prob{\mc{E}_\setS| \bar{\joutage}_\setS}\notag\\
&\doteq \snr^{-d_\setS(\gamma_\setS)}\label{tmp3}
\end{align}
where \eref{tmp3} is obtained by the same reasoning as used in the proof of Theorem \ref{Th.CDC} with the users' codebooks $\{\mc{C}_{r_u}:u\in\setS\}$ satisfying \eref{Eq.NewCDC} instead of \eref{Eq.ThCDC}. Inserting \eref{tmp2} and \eref{tmp3} into \eref{tmp1} yields
\begin{align}
P_e(\codebook_{\bm{r}}) &\dotleq \snr^{-d_{\setS^\star}(r(\setS^\star))} + \sum_{\setS\neq\setS^\star} \snr^{-d_\setS(\gamma_\setS)}\label{tmp4}\\
&\doteq \snr^{-d_{\setS^\star}(r(\setS^\star))}\label{tmp5}
\end{align}
where \eref{tmp5} follows from the fact that \eref{Eq.gamma} implies $d_\setS(\gamma_\setS) \geq  d_{\setS^\star}(r(\setS^\star))$, for all $\setS\neq \setS^\star$, and consequently, the dominant outage event dominates the upper bound on the total error probability. With $P_e(\codebook_{\bm{r}}) \dotgeq \prob{\outage}$ \cite[Lemma 7]{TVZ04}, combining \eref{Eq.LB} and \eref{tmp5} yields
\begin{equation}
P_e(\codebook_{\bm{r}}) \doteq \prob{\outage} \doteq \snr^{-d_{\setS^\star}(r(\setS^\star))}\label{tee1}.
\end{equation}
Since, by definition, $d(\codebook_{\bm{r}}) \leq d^\star(\bm{r})$, using \eref{Eq.OutLB}, we can finally conclude from \eref{tee1} that
\begin{equation}
d(\codebook_{\bm{r}}) = d^\star(\bm{r})= d_{\setS^\star}(r(\setS^\star)).\label{tee5}
\end{equation}
\end{proof}

As a consequence of Theorem \ref{Th.Proc}, the optimal DM tradeoff is determined by the tradeoff curve $d_{\setS^\star}(r(\setS^\star))$, which is simply the SNR exponent of the Jensen outage probability $\prob{\mathcal{J}_{\setS^\star}}$ corresponding to the dominant outage set. By virtue of \eref{Eq.Outprob2}, \eref{tee1}, and the fact that the relations $\prob{\outage_\setS}\leq\prob{\outage}$ and $\prob{\joutage_\setS}\leq\prob{\outage_\setS}$ hold for every $\setS$ and, a fortiori, for the dominant outage set $\setS^\star$, we get
\begin{equation}
\prob{\mathcal{O}_{\setS^\star}} \dotleq \prob{\mathcal{O}} \doteq \prob{\mathcal{J}_{\setS^\star}} \dotleq \prob{\mathcal{O}_{\setS^\star}} 
\end{equation}
which is to say that
\begin{equation}
\prob{\mathcal{O}_{\setS^\star}} \doteq \snr^{-d_{\setS^\star}(r(\setS^\star))}.
\end{equation}
Hence, as in the point-to-point case \cite{pcoj07}, the Jensen upper bound on mutual information yields a lower bound on the outage probability which is exponentially tight (in SNR).

In order to achieve DM tradeoff optimal performance, the families of codes $\{\codebook_{r_u}, u\in\mc{U}\}$ are required to satisfy \eref{Eq.ThCDC} for the dominant outage set $\setS^\star$ and, in addition, the probability $\prob{\mc{E}_{\setS}}$ corresponding to the sets $\setS\neq\setS^\star$ should decay at least as fast as $\prob{\outage_{\setS^\star}}=\prob{\joutage_{\setS^\star}}$, a requirement that is guaranteed when \eref{Eq.NewCDC} is satisfied for every $\setS\neq\setS^\star$. Note that this code design criterion is less stringent than requiring all the terms $\prob{\mc{E}_\setS}$ to satisfy condition \eref{Eq.ThCDC}, as originally proposed in \cite[Th. 2]{CGB08}. 
We conclude by pointing out that the code design criterion in Theorem \ref{Th.Proc} was shown to be necessary and sufficient for DM tradeoff optimality in Rayleigh flat-fading MACs in \cite{AB09}. We stress, however, that there exist codes---at least in the two-user flat-fading case---that satisfy \eref{Eq.ThCDC} in Theorem \ref{Th.CDC} for all $\setS\subseteq\mc{U}$ as we will show in Section \ref{Sec.CodeEx}.

\subsection{Dominant outage event regions}

The following example illustrates the application of Theorem \ref{Th.Proc} to the two-user case, and reveals the existence of multiplexing rate regions dominated by different outage events. Remarkably, although the error mechanism at play here (outage) is different from the one in \cite{Gallager85}, the dominant outage event regions we obtain have a striking resemblance to the dominant error event regions found in \cite{Gallager85}.

\subsubsection*{Example}
We assume $\mt=3$, $\mr=4$, and $\rank{\mat{R}_\mbb{H}}=\rho=2$. For $U=2$, the $2^2-1 = 3$ possible outage events are denoted by $\mc{O}_1$ (user 1 is in outage), $\mc{O}_2$ (user 2 is in outage) and $\mc{O}_3$ (the channel obtained by concatenating both users' channels into an equivalent point-to-point channel is in outage). The SNR exponents of the corresponding outage probabilities are obtained from \eref{Eq.JensenCurve} as 
\begin{align}
d_u(r_u) &= (3-r_u)(8-r_u), \quad u=1,2,\label{Eq.Ex}\\
d_3(r_1+r_2) &= \big(4-(r_1+r_2)\big)\big(12-(r_1+r_2)\big)\label{Eq.Ex2}.
\end{align}
Based on \eref{Eq.Ex} and \eref{Eq.Ex2}, we can now explicitly determine the dominant outage event for every multiplexing rate tuple $\bm{r}=(r_1, r_2)$. In Fig. \ref{Fig.TEE}, we plot the rate regions dominated by the different outage events. Note that the boundaries $r_1<3$, $r_2<3$, and $r_1+r_2 <4$ are determined by the ergodic capacity region. In the rate region dominated by $\mc{O}_1$, we have $d_1(r_1)<d_2(r_2)$ and $d_1(r_1)<d_3(r_1+r_2)$, implying that the SNR exponent of the total error probability equals $d_1(r_1)$, i.e., the SNR exponent that would be obtained in a point-to-point selective-fading MIMO channel with $\mt=3$, $\mr=4$, and $\rho=2$. The same reasoning applies to the rate region dominated by $\mc{O}_2$ and, hence, we can conclude that, in the sense of the DM tradeoff, the performance in regions $\mc{O}_1$ and $\mc{O}_2$ is not affected by the presence of the respective other user. In contrast, in the area dominated by $\mc{O}_3$, we have $d_3(r_1+r_2)<d_u(r_u)$, $u=1,2$, which is to say that multiuser interference does have an impact on the DM tradeoff and reduces the diversity order that would be obtained if only one user were present. 

\begin{figure}
\begin{center}
\includegraphics[width=.7\textwidth]{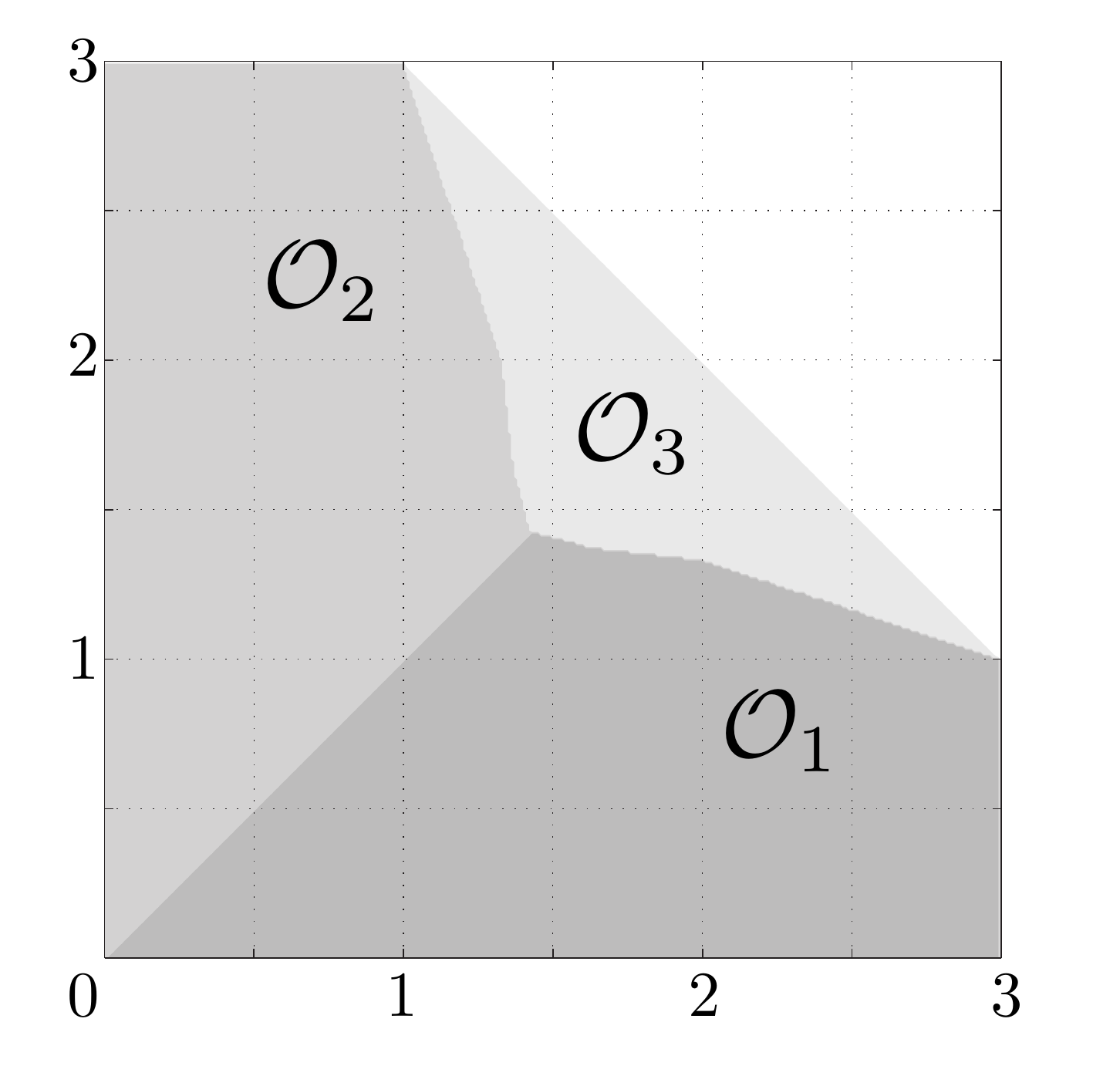}
\caption{\label{Fig.TEE} Dominant outage event regions for a two-user MA MIMO channel with $\mt=3$, $\mr=4$, and $\rho=2$.}
\end{center}
\end{figure}

Fig. \ref{Fig.TEE2} shows the dominant outage event regions for the same system parameters as above but with one additional receive antenna, i.e., $\mr=5$. We observe that not only larger sum multiplexing rates are achievable, i.e., $r_1+r_2 \leq 5$, but also that the area where $\mc{O}_3$ dominates the total error probability, and hence where multiuser interference reduces the achievable diversity order, is significantly smaller relative to the area dominated by the single user outage events $\mc{O}_1$ and $\mc{O}_2$. This effect can be attributed to the fact that increasing $\mr$ yields more spatial degrees of freedom at the receiver and, consequently, alleviates the task of resolving multiuser interference.

\begin{figure}
\begin{center}
\includegraphics[width=.7\textwidth]{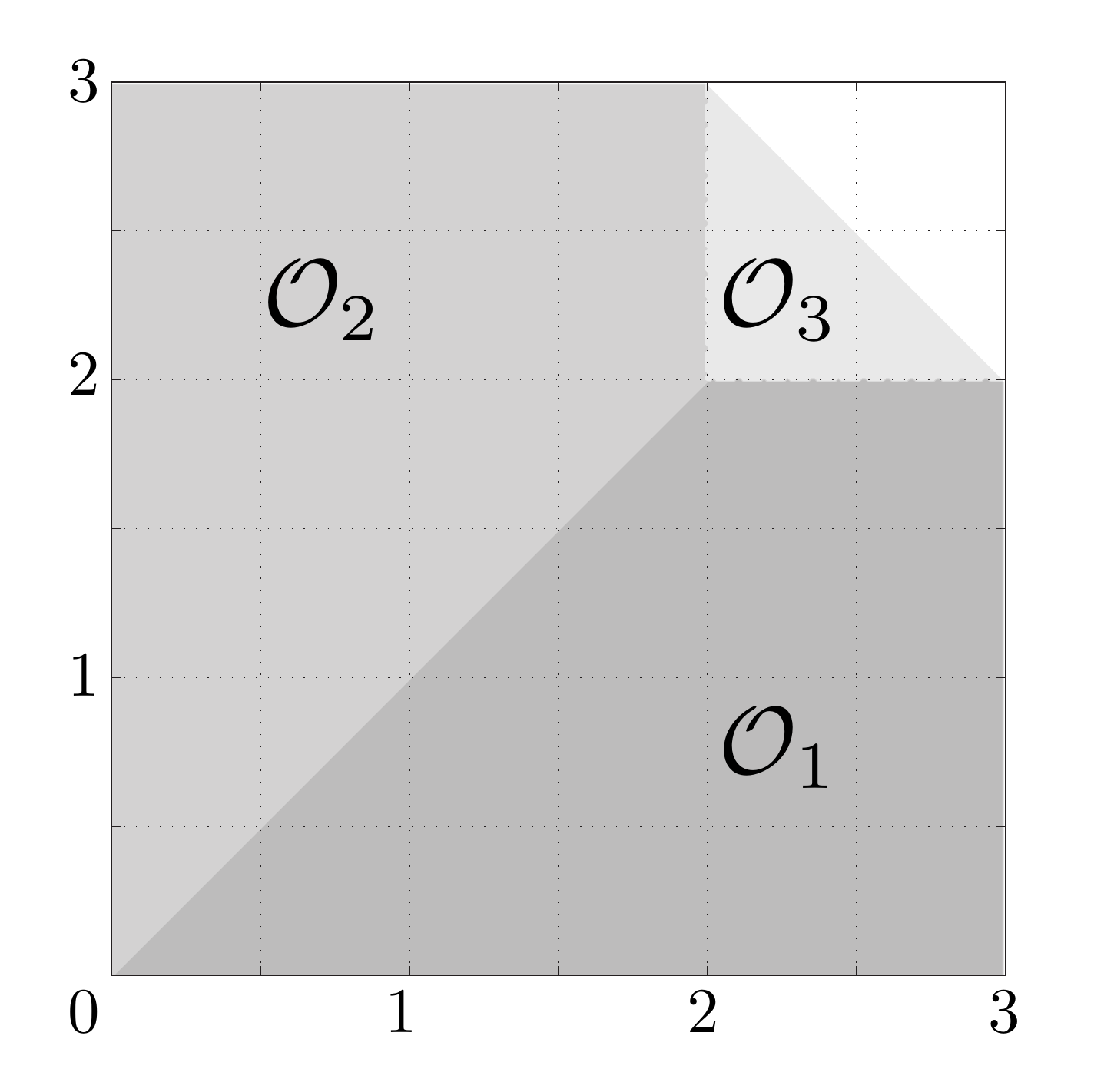}
\caption{\label{Fig.TEE2} Dominant outage event regions for a two-user MA MIMO channel with $\mt=3$, $\mr=5$, and $\rho=2$.}
\end{center}
\end{figure}

\subsection{Multiplexing rate region at a given diversity level}

The dominant outage event determines the maximum achievable diversity order as a function of the multiplexing rate tuple $\bm{r}$. Conversely, one may also be interested in finding the region $\mc{R}(d)$ of achievable multiplexing rates at a minimum diversity order $d\in[0,\rho \mt\mr]$ associated with the total error probability. This is accomplished by designing an overall family of codes that is DM tradeoff optimal and satisfies
\begin{equation}\label{cmp1}
d_{\setS}(r(\setS))\geq d,  \quad \forall \setS\subseteq\mc{U}
\end{equation}
which upon application of $r_{\scriptscriptstyle\setS}(\cdot)$ to both sides is found to be equivalent to
\begin{equation*}
r(\setS) \leq r_{\scriptscriptstyle\setS}(d), \quad \forall \setS\subseteq\mc{U}.
\end{equation*}
We just proved the following extension of \cite[Th. 2]{TVZ04} to selective-fading MA MIMO channels.

\begin{cor}\label{Cor.MRR}
Consider an overall family of codes $\codebook_{\bm{r}}$ that achieves the optimal DM tradeoff in the sense of Theorem \ref{Th.Proc}. Then, the region of multiplexing rates for which the total error probability decays with SNR exponent at least equal to $d$ is characterized by
\begin{equation}
\mc{R}(d) \triangleq \bigg\{\bm{r}:  r(\setS) \leq r_{\scriptscriptstyle\setS}(d), \forall \setS\subseteq\mc{U}\bigg\}\label{eq:MRR}
\end{equation}
where $r_{\scriptscriptstyle\setS}(d)$ is the inverse function of $d_\setS(r)$.
\end{cor}

To illustrate the concept of a multiplexing rate region \cite{TVZ04}, consider the two-user case with $\mt=3$, $\mr=4$, and $\rho=2$. Fig. \ref{Fig.mrr} shows the multiplexing rate regions $\mc{R}(d)$ corresponding to several diversity order levels, i.e., $d \in \{0,2,4,8,16\}$. The region $\mc{R}(0)$ is the pentagon described by the constraints $r_1\leq 3$, $r_2\leq 3$, and $r_1+r_2\leq \min(2\mt,\mr)=4$. Higher diversity order can be achieved at the expense of tighter constraints on the achievable multiplexing rates $r_1$ and $r_2$. For instance, for a diversity order requirement of $d\geq 8$, the achievable multiplexing rate region is given by the pentagon $0\mathrm{ABCD}$. Increasing the minimum required diversity order results in multiplexing rate regions that shrink towards the origin. Note that to realize a diversity order requirement of $d\geq 16$, the allowed multiplexing rate region is a square; in this case, performance (in the sense of the DM tradeoff) is not affected by the presence of a second user. Intuitively, the required diversity order is so high that users can only communicate at very small multiplexing rates and multiuser interference does not dominate the total error probability.

\begin{figure}
\begin{center}
\includegraphics[width=.7\textwidth]{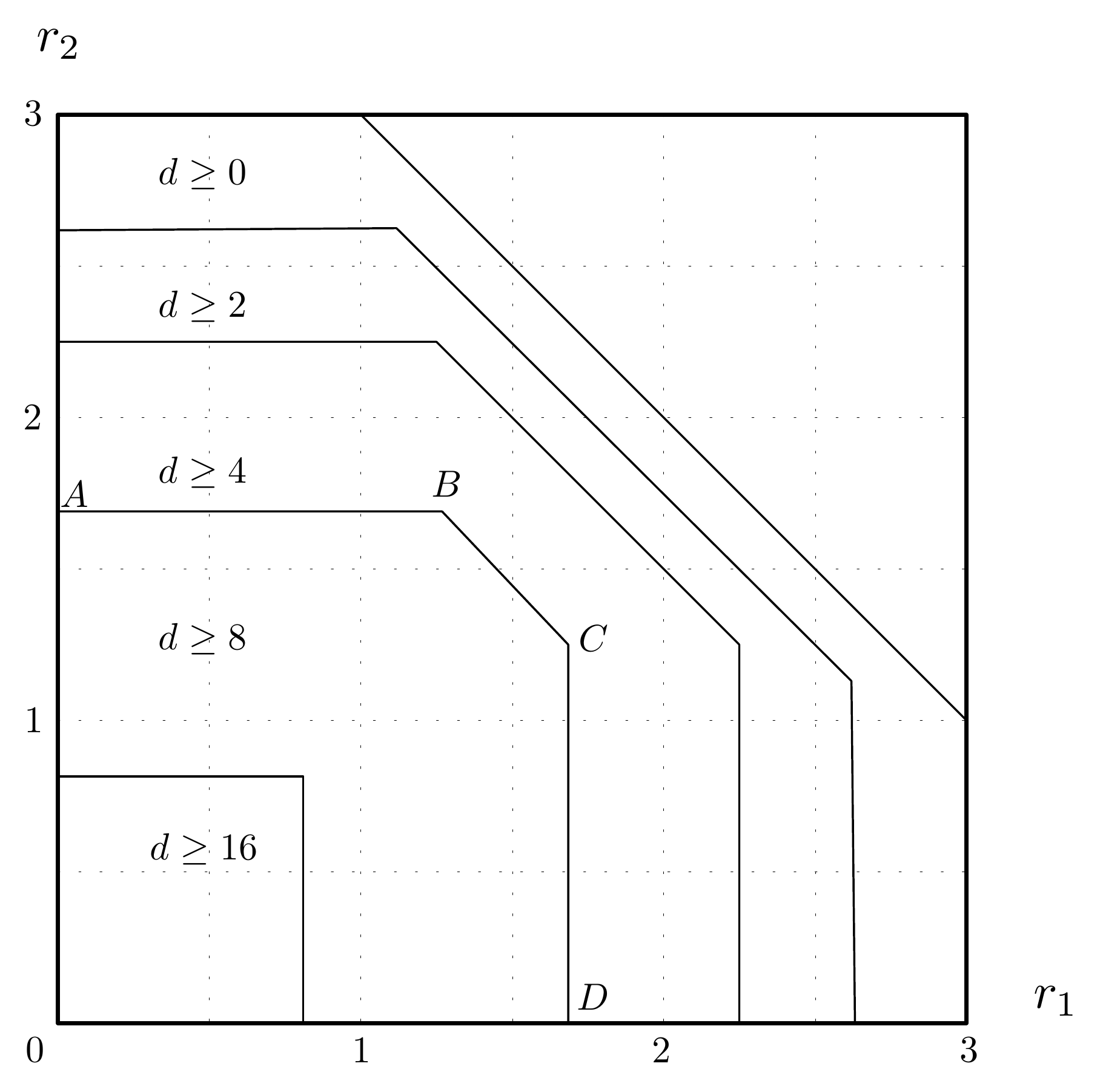}
\caption{\label{Fig.mrr} Multiplexing rate regions as a function of the diversity order $d\in\sizecurly{0,2,4,8,16}$ corresponding to the total error probability ($\mt=3$, $\mr=4$, and $\rho=2$).}
\end{center}
\end{figure}

\section{Analysis of a code construction for the two-user flat-fading case}\label{Sec.CodeEx}

In this section, we study the algebraic code construction proposed recently in \cite{BadBel08} for flat-fading MACs with 
two single-antenna users and an arbitrary number of antennas at the receiver. We examine whether this code satisfies the 
code design criteria of Theorem \ref{Th.Proc} and focus on the case of a two-antenna receiver, for simplicity.

We start by briefly reviewing the code construction described in \cite{BadBel08} for a system with $\mt=1$, $\mr=2$, $U=2$, $N=2$, and $\rho=1$ (i.e., flat fading). For each user $u$, let $\mc{A}_u$ denote a QAM constellation with $2^{R_u'(\snr)}$ points carved from $\mbb{Z}[i]=\{k+il:k,l\in\mbb{Z}\}$, where $i=\sqrt{-1}$ and $R_u'(\snr)=(r_u-\epsilon)\log\snr$ for some $\epsilon>0$, i.e.,
\begin{equation}\label{constellation}
\mc{A}_u = \sizecurly{(k+il): \frac{-2^{R_u'(\snr)/{2}}}{{2}}\leq k, l\leq\frac{2^{R_u'(\snr)/2}}{{2}}, \;k, l\in\mbb{Z}}.
\end{equation}
The proposed code spans two slots so that the vector of information symbols corresponding to user $u$ is given by $\mat{s}_u = [s_{u,1}\; s_{u,2}]$, where $s_{u,1}, s_{u,2} \in \mc{A}_u$. The vector $\mat{s}_u$ is then encoded using the unitary transformation matrix $\mat{U}$ underlying the Golden Code \cite{BelRekVit05} according to\looseness-1
\begin{equation}\label{Eq.GoldenTransform}
\tilde{\mat{x}}_u^T = \mat{U}\:\mat{s}_u^T = \begin{bmatrix}{x}_u \\ \sigma({x}_u) \end{bmatrix}\; \text{with }\mat{U}=\frac{1}{\sqrt{5}}\begin{bmatrix} \alpha & \alpha\varphi \\ \bar{\alpha} & \bar{\alpha}\bar{\varphi}\end{bmatrix}
\end{equation}
where $\varphi=\frac{1+\sqrt{5}}{2}$ denotes the Golden number with corresponding conjugate $\bar{\varphi}=\frac{1-\sqrt{5}}{2}$, $\alpha=1+i-i\varphi$ and $\bar{\alpha}=1+i-i\bar{\varphi}$. By construction, $x_u$ belongs to the quadratic extension $\mbb{Q}(i,\sqrt{5})$ over $\mbb{Q}(i)=\sizecurly{k+il : k,l\in \mbb{Q}}$. Here, $\sigma$ denotes the generator of the Galois group of $\mbb{Q}(i,\sqrt{5})$ given by\looseness-1
\begin{equation}\label{Eq.Gen}
\begin{array}{cccc} \sigma: & \mbb{Q}(i,\sqrt{5}) & \rightarrow & \mbb{Q}(i,\sqrt{5})\\
& a+ b \sqrt{5} & \mapsto & a -b\sqrt{5}.
\end{array}
\end{equation}
Moreover, one of the users, say user 2, multiplies the symbol transmitted in the first slot by a constant $\gamma\in \mbb{Q}(i)$, resulting in the transmit codeword
\begin{equation}\label{Eq.Codeword}
\tilde{\mat{X}} = \begin{bmatrix}{x}_{1}& \sigma({x}_{1}) \\ \gamma{x}_{2} & \sigma({x}_2) \end{bmatrix}.
\end{equation}

Depending on the choice of the parameter $\gamma$, the codeword difference matrices arising from this construction have a nonzero determinant. For completeness, we shall next provide a proof of this statement which was originally made in \cite{BadBel08}.

\begin{tc}\label{Th3}
For any $\gamma \neq \pm 1$ and any two $\tilde{\mat{X}}, \tilde{\mat{X}}'$ according to \eref{Eq.Codeword}, it holds that $\det(\mat{\Delta})\neq 0$, where $\mat{\Delta}=\tilde{\mat{X}}-\tilde{\mat{X}}'$.\looseness-1
\end{tc}

\begin{proof}
Proceeding along the same lines as \cite{BadBel08}, we start by proving that the determinant corresponding to any codeword $\tilde{\mat{X}}$ in \eref{Eq.Codeword} is nonzero for any $\gamma\neq\pm1$, and hence, by the linearity of the mapping $\sigma$ over $\mbb{Q}(i,\sqrt{5})$, the determinant of any codeword difference matrix is also nonzero. Note that 
\begin{eqnarray}
\det(\tilde{\mat{X}}) &= x_1\sigma(x_2)-\gamma x_2\sigma(x_1)\notag\\
&= x-\gamma \sigma(x) \label{rn}
\end{eqnarray}
where the last step follows from setting $x=x_1\sigma(x_2)$, noting that $\sigma(\sigma(x))=x$ for any $x\in\mbb{Q}(i,\sqrt{5})$, and using the property $\sigma(x\cdot y) = \sigma(x)\cdot\sigma(y)$ for every $x, y\in\mbb{Q}(i,\sqrt{5})$. Hence, $\det(\tilde{\mat{X}})$ is zero if and only if $\gamma$ satisfies $\gamma = x/\sigma(x)$. In this case, recalling that $\gamma \in \mbb{Q}(i)$, we must have  $x \in \mbb{Q}(i)$, or $x \in \sqrt{5}\mbb{Q}(i)=\sizecurly{\sqrt{5}(k+il) : k,l\in \mbb{Q}}$. These constraints yield, respectively, $\gamma=x/\sigma(x)= 1$ and $\gamma=x/\sigma(x) = -1$, from which we can infer that $\det(\tilde{\mat{X}})=0$ $\iff$ $\gamma=\pm 1$. Hence, any $\gamma\in\mbb{Q}(i)\backslash\{\pm1\}$ guarantees $\det(\tilde{\mat{X}})\neq 0$ for $x_1, x_2 \in \mbb{Q}(i,\sqrt{5})$.
\end{proof}

We are now ready to examine whether this construction satisfies the code design criteria for DM tradeoff optimality given in Theorem \ref{Th.Proc}. For simplicity, we assume $\gamma=i$ in the following.

We start by considering the cases $\setS=\{1\}$ and $\setS=\{2\}$. Assume that $\mc{A}_u$ is chosen according to \eref{constellation}. By \eref{Eq.GoldenTransform} and the fact that $\mat{U}$ is unitary, we obtain
\begin{align}
\max_{\begin{subarray}{c}\tilde{\mat{x}}_u\: : \: \tilde{\mat{x}}_u \: = \: {\mat{s}}_u\mat{U}^T \\ {s}_{u,1}, s_{u,2} \:\in\: \mc{A}_u\end{subarray}} \norm{\tilde{\mat{x}}_u}^2 &= \max_{{s}_{u,1}, s_{u,2} \:\in\: \mc{A}_u}\norm{{\mat{s}}_u}^2\notag\\
&=2\:\sizeparentheses{\frac{2^{R_u'(\snr)}}{2}}\label{cmp6}
\end{align}
for $u=1,2$. In order to satisfy the power constraint \eref{Eq.PC}, we scale the transmit vector corresponding to user $u$ as 
\begin{equation}\label{Eq.scaling}
\mat{x}_u =\sizeparentheses{\frac{2^{R_u'(\snr)}}{2}}^{-1/2} \:\tilde{\mat{x}}_u
\end{equation} 
so that, using \eref{cmp6}, we get
\begin{align}
\max_{\mat{x}_u \:\in\: \codebook_{r_u}(\snr)}\norm{\mat{x}_u}^2 &= \sizeparentheses{\frac{2^{R_u'(\snr)}}{2}}^{-1} \max_{\begin{subarray}{c}\tilde{\mat{x}}_u\: : \: \tilde{\mat{x}}_u \: = \: {\mat{s}}_u\mat{U}^T \\ {s}_{u,1}, s_{u,2} \:\in\: \mc{A}_u\end{subarray}} \norm{\tilde{\mat{x}}_u}^2\notag\\
&= 2.\label{cmp2}
\end{align}
For user 2, we note that \eref{cmp2} remains valid after multiplying the first entry of $\mat{x}_2$ by $\gamma=i$. Next, we note that in the flat-fading case $\mat{R}_\mbb{H}^T\odot(\mat{e}_u^H\mat{e}_u)= \mat{e}_u^H\mat{e}_u$, where $\mat{e}_u=\mat{x}_u-\mat{x}_u'$ for $\mat{x}_u,\mat{x}_u' \in \codebook_{r_u}(\snr)$ and $u=1,2$. Considering user 1, i.e., $\setS=\{1\}$, we have $|\setS|=1$ and $\minant(\setS)=1$ so that the quantity defined in \eref{Eq.DefLambda} is simply the smallest squared norm of the first row in \eref{Eq.MatE} and satisfies
\begin{align}
\Lambda_1^{1}(\snr) &= \min_{\begin{subarray}{c} \mat{x}_1,\mat{x}_1' \: \in\: \codebook_{r_1}(\snr)\end{subarray}} \quad \norm{\mat{x}_1-\mat{x}_1'}^2\notag\\
&= 2^{1-R_1'(\snr)}\:\min_{\begin{subarray}{c}\tilde{\mat{x}}_1\: : \: \tilde{\mat{x}}_1 \: = \: {\mat{s}}_1\mat{U}^T\:;\: \tilde{\mat{x}}_1'\: : \: \tilde{\mat{x}}_1' \: = \: {\mat{s}}_1'\mat{U}^T\\{s}_{1,1}, \: s_{1,2}, \:{s}_{1,1}', \: s_{1,2}'\:\in\: \mc{A}_1\end{subarray}} \quad \norm{\tilde{\mat{x}}_1-\tilde{\mat{x}}_1'}^2\label{cmp4}\\
&= 2^{1-R_1'(\snr)}\: \min_{{s}_{1,1}, \: s_{1,2}, \:{s}_{1,1}', \: s_{1,2}'\:\in\: \mc{A}_1}  \underbrace{||\mat{s}_1-\mat{s}_1'||^2}_{\geq \:d_{\min}^2}\label{cmp5}
\end{align}
where \eref{cmp4} follows from \eref{Eq.scaling}, and \eref{cmp5} is a consequence of $\tilde{\mat{x}}_1^T = \mat{U}\mat{s}_1^T$ and the unitarity of $\mat{U}$. From \eref{constellation}, we note that $d_{\min}=1$, i.e., the minimum distance in $\mc{A}_1$ is independent of SNR, and invoking $R_1'(\snr)=(r_1-\epsilon) \log\snr$, we can conclude from \eref{cmp5} that 
\begin{equation*}
\Lambda_1^{1}(\snr)\doteq \snr^{-(r_1-\epsilon)}. 
\end{equation*}
For user 2, a similar argument\footnote{The multiplication of the first component of $\tilde{\mat{x}}_2$ by $\gamma=i$ does not affect the Euclidean norm.} shows that $\Lambda_1^{1}(\snr)\doteq \snr^{-(r_2-\epsilon)}$ and, hence, the code satisfies the criteria arising from \eref{Eq.ThCDC} for $\setS=\{1\}$ and $\setS=\{2\}$. 

Next, we consider the case $\setS=\{1,2\}$. The overall transmit codeword is now given by
\begin{equation}\label{Eq.TxCodeword}
\mat{X}= \sqrt{2} \begin{bmatrix} 2^{-R_1'(\snr)/2} \:{x}_{1}& 2^{-R_1'(\snr)/2} \:\sigma({x}_{1}) \\ 2^{-R_2'(\snr)/2} \:i{x}_{2} & 2^{-R_2'(\snr)/2} \:\sigma({x}_2) \end{bmatrix}
\end{equation}
and satisfies the power constraint \eref{Eq.PC}, i.e.,
\begin{align*}
\max_{\mat{X}\:\in\:\codebook_{\bm{r}}(\snr)}\norm{\mat{X}}_{\mathrm{F}}^2 &= \max_{\mat{X}\:\in\:\codebook_{\bm{r}}(\snr)}\trace{\mat{XX}^H}\\
&=\max_{\mat{x}_1 \:\in\: \codebook_{r_1}(\snr)}\norm{\mat{x}_1}^2+\max_{\mat{x}_2 \:\in\: \codebook_{r_2}(\snr)}\norm{\mat{x}_2}^2\\
&= 4.
\end{align*}
From \eref{Eq.TxCodeword} and the linearity of the mapping $\sigma$ over $\mbb{Q}(i,\sqrt{5})$, the codeword difference matrix is obtained as\looseness-1
\begin{equation}\label{Eq.MatE}
\mat{E}= \sqrt{2} \begin{bmatrix} 2^{-R_1'(\snr)/2} \:{e}_{1}& 2^{-R_1'(\snr)/2} \:\sigma({e}_{1}) \\ 2^{-R_2'(\snr)/2} \:i {e}_{2} & 2^{-R_2'(\snr)/2} \:\sigma({e}_2) \end{bmatrix}
\end{equation}
where $e_u = x_u-x_u'$ and hence $e_u \in \mbb{Q}(i,\sqrt{5})$, $u=1,2$. Recall that in the flat-fading case $\mat{R}_\mbb{H}^T\odot(\mat{E}^H\mat{E})= \mat{E}^H\mat{E}$. Next, note that $|\setS|=2$ and $\minant(\setS)=2$ so that $\Lambda_2^{2}(\snr)=\min_\mat{E}\abs{\det(\mat{E})}^2$. From \eref{Eq.MatE}, we readily get 
\begin{equation}\label{new1}
\min_{\begin{subarray}{c}\mat{E}={\mat{X}}-{\mat{X}}'\\ \mat{X},\mat{X}'\:\in\:\codebook_{\bm{r}}(\snr)\end{subarray}}\abs{\det(\mat{E})}^2 =2^{1-(R_1'(\snr)+R_2'(\snr))} \:\min_{\begin{subarray}{c}\mat{\Delta}=\tilde{\mat{X}}-\tilde{\mat{X}}'\\ \tilde{\mat{X}} = f(\mat{s}_1, \mat{s}_2)\:,\:\tilde{\mat{X}}' = f(\mat{s}_1', \mat{s}_2')\\ s_{u,1}, s_{u,2}, s_{u,1}', s_{u,2}' \:\in\: \mc{A}_u,\: u=1,2\end{subarray}}\abs{\det(\mat{\Delta})}^2
\end{equation} 
where we have used the notation $\tilde{\mat{X}} = f(\mat{s}_1, \mat{s}_2)$ to express the fact that $\tilde{\mat{X}}$ is obtained from $\mat{s}_1$ and $\mat{s}_2$ using \eref{Eq.GoldenTransform} and \eref{Eq.Codeword}. We recall that $\det(\mat{\Delta})\neq 0$ for $\mat{\Delta}$ arising from any combination of vectors $\mat{s}_u, \mat{s}_u'$ ($u=1,2$) with entries in $\mbb{Z}[i]$. Therefore, for every SNR, there must exist an $\omega(\snr)>0$ such that 
\begin{equation}\label{new}
\min_{\begin{subarray}{c}\mat{\Delta}=\tilde{\mat{X}}-\tilde{\mat{X}}'\\ \tilde{\mat{X}} = f(\mat{s}_1, \mat{s}_2)\:,\:\tilde{\mat{X}}' = f(\mat{s}_1', \mat{s}_2')\\ s_{u,1}, s_{u,2}, s_{u,1}', s_{u,2}' \:\in\: \mc{A}_u,\: u=1,2\end{subarray}}\abs{\det(\mat{\Delta})}^2 = \omega(\snr)
\end{equation}
which, upon inserting into \eref{new1} and using $R_u'(\snr)=(r_u-\epsilon) \log\snr$ ($u=1,2$), yields
\begin{align}
\Lambda_2^{2}(\snr)&=\min_{\begin{subarray}{c}\mat{E}={\mat{X}}-{\mat{X}}'\\ \mat{X},\mat{X}'\:\in\:\codebook_{\bm{r}}(\snr)\end{subarray}}\abs{\det(\mat{E})}^2 \notag\\
&\doteq  \snr^{-(r_1+r_2-2\epsilon)} \:\omega(\snr).\label{new2}
\end{align}
It follows from (\ref{new})\,---by inspection---\,that $\omega(\snr)$ is a nonincreasing function of $\snr$. Unfortunately, Theorem~\ref{Th3} does not allow us to conclude that $\omega(\snr)$ is bounded away from zero, in which
case we could conclude from \eref{new2} and \eref{Eq.ThCDC} that the code is DMT-optimal for all multiplexing rate tuples. Therefore, characterizing the decay of $\omega(\snr)$ as a function of SNR is key to proving or disproving 
the DM tradeoff optimality of the code construction. Unfortunately, we have not been able to determine how $\omega(\snr)$ decays with SNR\footnote{We would like to use this chance to point out that despite the claim we made in \cite{CGB08}, we do not have a proof establishing the DM tradeoff optimality of the code construction in \cite{BadBel08}.}. Characterizing this decay rate seems very difficult and is likely to require advanced algebraic concepts. We can, however, distinguish between three different possibilities. If $\omega(\snr)$ decays exponentially with SNR, the criteria for DM tradeoff optimality provided in this paper are not met. In the case of a subpolynomial decay, i.e., 
\begin{equation*}
\lim_{\snr\rightarrow\infty} \frac{\log \omega(\snr)}{\log\snr}=0
\end{equation*}
we would get $\Lambda_2^{2}(\snr) \doteq \snr^{-(r_1+r_2-2\epsilon)}$ and, hence, such a decay would be sufficient to guarantee that \eref{new2} satisfies the code design criterion \eref{Eq.ThCDC} for $\setS=\{1,2\}$ and any tuple $(r_1,r_2)$ in the multiplexing rate region. Finally, we consider the case of $\omega(\snr)$ exhibiting polynomial decay, assuming that $\omega(\snr) \doteq \snr^{-\delta}$, $\delta>0$. In this case, it would follow from \eref{new2} that
\begin{equation*}
\Lambda_2^{2}(\snr) \doteq\snr^{-(r_1+r_2+\delta-2\epsilon)}.
\end{equation*}
The quantity $\Lambda_2^{2}(\snr)$ would then decay faster than required by \eref{Eq.ThCDC}. In other words, the code construction would not be DM tradeoff optimal in the sense of Theorem \ref{Th.Proc} when the dominant outage set is $\setS^\star=\{1,2\}$. However, when the dominant outage set is either $\setS^\star=\{1\}$ or $\setS^\star=\{2\}$, the relaxed (compared to the code design criteria proposed in \cite{CGB08}) code design criteria provided in \eref{Eq.NewCDC} would still be met for any multiplexing rate tuple $(r_1, r_2)$ satisfying
\begin{equation*}
r_1+r_2+\delta \leq r_\setS(d_{\setS^\star}(r(\setS^\star))).
\end{equation*}



We conclude this section by noting that a DM tradeoff optimal code construction for flat-fading MACs was reported in \cite{NamGam07}. Specifically, it is shown in \cite{NamGam07} that lattice-based space-time codes achieve the optimal DM tradeoff with lattice decoding. As a consequence of the code design criterion in Theorem \ref{Th.Proc} being necessary and sufficient for DM tradeoff optimality 
in Rayleigh flat-fading MACs \cite{AB09}, the code construction reported in \cite{NamGam07} necessarily satisfies these design criteria. The systematic construction of DM tradeoff optimal codes for selective-fading MA MIMO channels seems, however, largely unexplored.

\section{Conclusion\label{Sec.Conclusion}}

We characterized the optimum DM tradeoff for selective-fading MA MIMO channels and studied corresponding code design criteria. Our results show that, for a prescribed multiplexing rate tuple, the optimal DM tradeoff is determined by the dominant outage event. The systematic design of DM tradeoff optimal codes for the (selective-fading) MIMO MAC remains an important open problem.

\section*{Acknowledgment}
The authors would like to thank C.~Ak\c caba for insightful remarks on the code design criterion in Theorem \ref{Th.Proc} and for stimulating discussions on its proof. 
We would furthermore like to thank Prof. E.~Viterbo for pointing out the SNR-dependency of $\omega(\snr)$ in \eref{new} and, in particular, as a consequence thereof a problem with the arguments used to arrive at the statement of the code proposed in \cite{BadBel08} satisfying our code design criteria for DM tradeoff optimality. Finally, we would like to thank Prof. J.~C.~Belfiore for very helpful discussions.

\bibliographystyle{IEEEtran}
\bibliography{biblio}

\begin{thebibliography}{10}
\providecommand{\url}[1]{#1}
\csname url@rmstyle\endcsname
\providecommand{\newblock}{\relax}
\providecommand{\bibinfo}[2]{#2}
\providecommand\BIBentrySTDinterwordspacing{\spaceskip=0pt\relax}
\providecommand\BIBentryALTinterwordstretchfactor{4}
\providecommand\BIBentryALTinterwordspacing{\spaceskip=\fontdimen2\font plus
\BIBentryALTinterwordstretchfactor\fontdimen3\font minus
  \fontdimen4\font\relax}
\providecommand\BIBforeignlanguage[2]{{%
\expandafter\ifx\csname l@#1\endcsname\relax
\typeout{** WARNING: IEEEtran.bst: No hyphenation pattern has been}%
\typeout{** loaded for the language `#1'. Using the pattern for}%
\typeout{** the default language instead.}%
\else
\language=\csname l@#1\endcsname
\fi
#2}}

\bibitem{ZheTse02}
L.~Zheng and D.~N.~C. Tse, ``Diversity and multiplexing: A fundamental tradeoff
  in multiple antenna channels,'' \emph{IEEE Trans. Inf. Theory}, vol.~49,
  no.~5, pp. 1073--1096, May 2003.

\bibitem{TVZ04}
D.~N.~C. Tse, P.~Viswanath, and L.~Zheng, ``Diversity-multiplexing tradeoff in
  multiple-access channels,'' \emph{{IEEE} Trans. Inf. Theory}, vol.~50, no.~9,
  pp. 1859--1874, Sep. 2004.

\bibitem{YaoWor03}
H.~Yao and G.~W. Wornell, ``Achieving the full {MIMO} diversity-multiplexing
  frontier with rotation based space-time codes,'' in \emph{Proc. Allerton
  Conf. on Commun., Control, and Computing}, Monticello, IL, Oct. 2003, pp.
  400--409.

\bibitem{GamCaiDam04}
H.~{El Gamal}, G.~Caire, and M.~O. Damen, ``Lattice coding and decoding
  achieves the optimal diversity-multiplexing tradeoff of {MIMO} channels,''
  \emph{{IEEE} Trans. Inf. Theory}, vol.~50, no.~9, pp. 968--985, Sep. 2004.

\bibitem{BelRekVit05}
J.-C. Belfiore, G.~Rekaya, and E.~Viterbo, ``The {G}olden code: {A} $2 \times
  2$ full rate space-time code with nonvanishing determinants,'' \emph{IEEE
  Trans. Inf. Theory}, vol.~51, no.~4, pp. 1432--1436, Apr. 2005.

\bibitem{TW05}
S.~Tavildar and P.~Viswanath, ``Approximately universal codes over slow-fading
  channels,'' \emph{IEEE Trans. Inf. Theory}, vol.~52, no.~7, pp. 3233--3258,
  Jul. 2006.

\bibitem{pco07}
P.~Coronel and H.~B{\"o}lcskei, ``Diversity-multiplexing tradeoff in
  selective-fading {MIMO} channels,'' in \emph{Proc. {IEEE} Int. Symp. Inf.
  Theory (ISIT)}, Nice, France, Jun. 2007, pp. 2841--2845.

\bibitem{NamGam07}
Y.~Nam and H.~{El Gamal}, ``On the optimality of lattice coding and decoding in
  multiple access channels,'' in \emph{Proc. {IEEE} Int. Symp. Inf. Theory
  (ISIT)}, Nice, France, Jun. 2007, pp. 211--215.

\bibitem{Gallager85}
R.~G. Gallager, ``A perspective on multiaccess channels,'' \emph{{IEEE} Trans.
  Inf. Theory}, vol.~31, no.~2, pp. 124--142, Mar. 1985.

\bibitem{BadBel08}
M.~Badr and J.-C. Belfiore, ``Distributed space-time block codes for the non
  cooperative multiple access channel,'' in \emph{Proc. Int. Zurich Seminar on
  Commun.}, Mar. 2008, pp. 132--135.

\bibitem{Bello63}
P.~A. Bello, ``Characterization of randomly time-variant linear channels,''
  \emph{IEEE Trans. Commun. Syst.}, vol. COM-11, pp. 360--393, 1963.

\bibitem{pcoj07}
P.~Coronel and H.~B{\"o}lcskei, ``Diversity-multiplexing tradeoff in
  selective-fading {MIMO} channels,'' in preparation.

\bibitem{CGB08}
P.~Coronel, M.~G\"artner, and H.~B{\"o}lcskei, ``Diversity-multiplexing
  tradeoff in selective-fading multiple-access {MIMO} channels,'' in
  \emph{Proc. {IEEE} Int. Symp. Inf. Theory (ISIT)}, Toronto, ON, Canada, Jul.
  2008, pp. pp. 915--919.

\bibitem{AB09}
C.~Ak{\c c}aba, ``Diversity-multiplexing tradeoff in relay and interference
  channels,'' {ETH} Dissertation, 2009.

\end{thebibliography}

\end{document}